\documentclass[11pt]{article}

\usepackage{amsmath}
\usepackage{amssymb}
\usepackage{float}
\usepackage{booktabs}
\usepackage{makecell}
\usepackage{bigints}
\usepackage{bbm}
\usepackage{caption}
\usepackage{multirow}
\usepackage{amsthm}
\usepackage{array}
\usepackage{bbm}
\usepackage{eucal}[mathscr]
\usepackage{textcomp}
\usepackage{tikz}
\usepackage{tkz-euclide}
\usetikzlibrary{angles,quotes}
\usepackage{graphicx}
\usepackage{hyperref}
\usepackage{pgfplots}
\pgfplotsset{compat=newest}
\usepackage{mathrsfs}
\usepackage[english]{babel}
\usepackage[autostyle, english = american]{csquotes}
\MakeOuterQuote{"}
\newcommand{\indep}{\perp \!\!\! \perp}
\usepgfplotslibrary{patchplots}

\DeclareMathOperator*{\argmin}{arg\,min}

\newtheorem{theorem}{Theorem}
\newtheorem{assumption}{Assumption}
\newtheorem{lemma}{Lemma}
\usepackage[backend=biber,
style=numeric,
citestyle=authoryear]{biblatex}

\title{Global Testing in Multivariate Regression Discontinuity Designs\footnote{I thank EunYi Chung, Marcelo Medeiros, Joshua Shea, and all of the participants of the UIUC Econometrics Lab Lunch Seminar and Applied Micro Research Lunch for thoughtful comments and suggestions.}}
\author{Artem Samiahulin\footnote{Department of Economics, University of Illinois at Urbana-Champaign}}
\date{\today}

\begin{document}
\maketitle

\begin{abstract}
Regression discontinuity (RD) designs with multiple running variables arise in a growing number of empirical applications, including geographic boundaries and multi-score assignment rules. Although recent methodological work has extended estimation and inference tools to multivariate settings, far less attention has been devoted to developing global testing methods that formally assess whether a discontinuity exists anywhere along a multivariate treatment boundary. Existing approaches perform well in large samples, but can exhibit severe size distortions in moderate or small samples due to the sparsity of observations near any particular boundary point. This paper introduces a complementary global testing procedure that mitigates the small-sample weaknesses of existing multivariate RD methods by integrating multivariate machine learning estimators with a distance-based aggregation strategy, yielding a test statistic that remains reliable with limited data. Simulations demonstrate that the proposed method maintains near-nominal size and strong power, including in settings where standard multivariate estimators break down. The procedure is applied to an empirical setting to demonstrate its implementation and to illustrate how it can complement existing multivariate RD estimators.
\end{abstract}

\pagebreak

\section{Introduction}

Regression discontinuity (RD) designs are widely regarded as one of the "...leading quasi-experimental empirical strategies in economics, political science, education, and many other social and behavioral sciences" (\textcite{calonico2014robust}). This paper develops a global testing framework for multivariate regression discontinuity designs, allowing researchers to determine whether discontinuities in conditional expectations or densities exist anywhere along a multivariate treatment boundary, even in settings with limited sample sizes. From a theoretical perspective, the paper establishes a global testing procedure whose validity does not rely heavily on precise estimation of local multivariate discontinuities. Practically, the proposed procedure exhibits good finite-sample size control even when the sample size is modest. \\

While traditional RD designs use a single running variable, recent research has increasingly explored settings with multiple running variables (\textcite{papay2011extending, reardon2012regression, wong2013analyzing}). For example, \textcite{keele2015geographic} examine geographic RD designs, where treatment and control groups are separated by geographic boundaries in two-dimensional space. \textcite{matsudaira2008mandatory} studies the effects of mandatory summer school, where treatment assignment depends on both math and reading scores. \textcite{londono2020upstream} explore a government tuition subsidy in Colombia where eligibility was determined by both merit (via a test score) and economic need (via an index of wealth). For more applications, see papers such as \textcite{frey2019cash, narita2021algorithm, elacqua2016short, egger2015impact, evans2017smart, cohodes2014merit, becht2016does}. \\

In response to these empirical applications, several methodological papers have been written to estimate treatment effects in multivariate RD settings. Notable contributions include \textcite{imbens2019optimized}, who use convex optimization, \textcite{cheng2023estimation}, who applies thin plate splines, and \textcite{gunsilius2023free}, who develop methods for cases where the treatment boundary is unknown. More recently, \textcite{cattaneo2025estimationdist}  and \textcite{cattaneo2025estimation} explore using local polynomial regressions to estimate treatment effects along two-dimensional boundaries, as well as aggregated effects. \\

Despite these methodological developments, relatively little attention has been devoted to global testing: the task of determining whether a discontinuity exists at any point along a multivariate treatment boundary. A straightforward strategy is to estimate treatment effects at many boundary locations and construct uniform confidence intervals at each point, assessing whether zero lies outside any of the resulting bands. When the sample size is large, this approach is appealing: it not only detects whether a discontinuity is present but also reveals where along the boundary it occurs. As a result, uniform confidence bands provide a rich and granular view of underlying treatment effect patterns. However, their performance deteriorates when the number of running variables increases or when the sample size is modest. In such settings, the number of observations effectively informing each boundary point becomes small, leading to unreliable asymptotic approximations. Through a simple simulation in the next subsection, this standard approach is shown to exhibit poor finite-sample performance. \\

To combat this issue, this paper proposes a new approach to global testing where a multivariate estimator is used to pool individual treatment effects together, which allows for more robust global inference even with smaller sample sizes. In simulations, a test size closer to a target 5\% is achieved even with a modest sample size. Even with two or more running variables, this approach can be used to test for discontinuities in multivariate conditional expectation functions and multivariate densities. This method is intended to complement existing multivariate RD estimators by providing a robust tool for assessing the reliability of estimated effects and confidence intervals. \\

The proposed global testing procedure works in two stages. First, a multivariate estimator is used to estimate whether the treatment effects along the boundary are positive or negative. The local linear forest estimator of \textcite{friedberg2020local} is used to estimate treatment effect heterogeneity, while a variation of the random forest density estimator of \textcite{wen2022random} is used for the density discontinuity test. Afterwards, these estimates are incorporated into a univariate test statistic that is used in the final hypothesis test. This approach has a number of advantages. First, by pooling together estimates along the boundary, the asymptotics for the final test statistic are made more reliable than they would be for each estimate individually. Second, the random forest estimators help maintain reasonable performance even when the number of running variables becomes relatively large. Third, the final test statistic only uses the signs of the multivariate estimates rather than the complete point estimates. This makes the final test statistic robust even when the initial treatment effects are estimated poorly, which translates to powerful theoretical properties. \\

This paper makes two contributions. First, it contributes to the literature on multivariate regression discontinuity designs by proposing a global test that remains robust even in relatively small samples. Second, it contributes to the literature on bunching and manipulation testing in regression discontinuity designs. Several papers have been written on univariate manipulation tests (\textcite{cattaneo2020simple, mccrary2008manipulation, cattaneo2024local}). \textcite{crippa2025manipulation} extends these types of tests to accommodate many running variables. However, this test only applies to boundaries like the one in Figure 1 in the next subsection. The test outlined in this paper can be applied to a wide range of boundary shapes, making it applicable even in settings with complex treatment rules, as in geographic regression discontinuities. \\

The remainder of this paper is organized as follows. Subsection 1.1 discusses a simulation example that demonstrates where the uniform confidence bands approach fails, while subsection 1.2 discusses notation and assumptions. Section 2 discusses identification for both the conditional expectation and density cases. Section 3 discusses estimation and inference for the final test statistics used for global testing. Section 4 shows simulation results that demonstrate the properties of the outlined global testing procedure. Section 5 discusses an application, and section 6 concludes.

\subsection{Simple Simulation Example}
To illustrate the limitations of the uniform confidence band approach in finite samples, consider the following process:
$$
Y_i = \frac{X_{1,i} + X_{2,i}}{3} + \epsilon_i
$$
where $\epsilon_i \sim \mathcal{N}(0, 0.05)$, $X_{1,i}, X_{2,i}$ come from independent uniform distributions on the interval $[-1, 1]$, and all observations are i.i.d across $i$. Further, define the treatment indicator as follows:
$$
T_i = \begin{cases}
    1, \quad X_{1,i}, X_{2,i} \geq 0 \\
    0, \quad \text{otherwise}
\end{cases}
$$
For this treatment rule, the treatment region looks as follows:
\begin{figure}[H]
\centering
\begin{tikzpicture}[scale = 3]
    % Define the regions with shading
    \fill[red, opacity=0.5] (-1,-1) rectangle (1,0);
    \fill[red, opacity=0.5] (-1,0) rectangle (0,1);
    \fill[blue, opacity=0.5] (0,0) rectangle (1,1);
    
    % Draw the axes
    \draw[thick,->] (-1.2,0) -- (1.2,0) node[right] {$x_1$};
    \draw[thick,->] (0,-1.2) -- (0,1.2) node[above] {$x_2$};

    % Add labels for clarity
    \node at (0.5,0.5) {Treated};
    \node at (-0.5,0.5) {Not Treated};
    \node at (0.5,-0.5) {Not Treated};
    \node at (-0.5,-0.5) {Not Treated};

\end{tikzpicture}
    \caption{Treatment and control regions for simulated treatment rule.}
\end{figure}
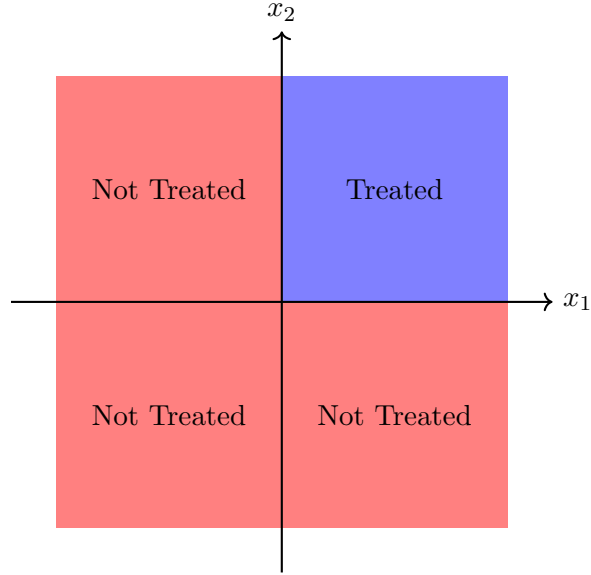

The goal is to test whether there is a discontinuity in $\mathbb{E}[Y_i | X_i = x]$ for some $x$ along the boundary between the treated and non-treated regions. To do this, the local polynomial estimator of \textcite{cattaneo2025estimation} is used to estimate treatment effect at 79 evenly-spaced points along the boundary and apply 95\% uniform confidence bands for inference. Afterwards, rejection rates are computed over 5,000 simulations. Since there are no discontinuities in this example, the rejection rate should be approximately 5\%. To assess robustness, the rejection rates from the sum of treatment effects and sum of squared treatment effects are also reported. Additionally, these three approaches are performed only on "safe" points $x = (0.5, 0)$ and $x = (0, 0.5)$ to avoid estimation at edges and kink points. Using number of observations $n = 20,000$ and $n = 1,000$, the rejection rates are as follows:

\begin{table}[H]
\centering
\caption{Simulated Rejection Rates for Bivariate Local Polynomial Regression}
\label{tab:rd2d-sim}
\begin{tabular}{c cc cc cc}
\toprule
& \multicolumn{2}{c}{Uniform Bands} & \multicolumn{2}{c}{Sum of $\hat{\tau}(x)$} & \multicolumn{2}{c}{Sum of $\hat{\tau}(x)^2$} \\
\cmidrule(lr){2-3} \cmidrule(lr){4-5} \cmidrule(lr){6-7}
$n$ & Full & Safe & Full & Safe & Full & Safe \\
\midrule
1,000  & 0.3652 & 0.1540 & 0.1006 & 0.1080 & 0.1714 & 0.1276 \\
20,000 & 0.0748 & 0.0560 & 0.0564 & 0.0566 & 0.0616 & 0.0598 \\
\bottomrule 
\end{tabular}
\end{table}

In the case where $n = 20,000$, the rejection rates are fairly close to 0.05. However, with a smaller number of observations, the rejection rate becomes much larger than the expected rejection rate of 0.05. Summation mitigates this problem only partially, because the aggregated statistic remains sensitive to inaccuracies in the individual local estimates used to construct it.

\subsection{Setup and Notation}
Let $Y_i$ denote the outcome variable for observation $i$ and let $X_i$ be a vector of running variables with support $\mathbb{X}$. Let $\Omega$ denote the set of treated units and let $\partial \Omega$ be the boundary of $\Omega$. Let $T_i = 1$ for $i \in \Omega$ and $T_i = 0$ for $i \notin \Omega$. Also, let $Y_i(t)$ denote the potential outcome for unit $i$ with treatment status $t$. Let $g_\Omega(X_i)$ be the signed (Euclidean) distance function relative to $\partial \Omega$, with $f_{g_\Omega}(g)$ being the density of $g_\Omega(X_i)$. For simplicity, define $g_\Omega(X_i) \equiv g_i$. With that, here are some necessary assumptions: 
\begin{assumption}
The set $\Omega$ is closed (without loss of generality), bounded, and has a piecewise smooth boundary.
\end{assumption}

\begin{assumption}
All observations are independently and identically distributed (i.i.d) across $i$.
\end{assumption}

\begin{assumption}
All $Y_i, X_i$ take values in bounded supports.
\end{assumption}

\begin{assumption}
$\mathbb{V}(Y_i(1) | X_i = x), \mathbb{V}(Y_i(0) | X_i = x) > 0$ are continuous and bounded for all $x \in \mathbb{X}$ within some neighborhood of $\partial \Omega$.
\end{assumption}

\begin{assumption}
The density $f_{g_\Omega}(g)$ is bounded away from zero in some neighborhood around $g = 0$.
\end{assumption}

\begin{assumption}
The kernel function $K(\cdot)$ is a bounded and symmetric function that takes values on the interval $[-1, 1]$.
\end{assumption}

\begin{assumption}
$\mathbb{E}[Y_i(t) | X_i = x]$ is twice continuously differentiable and bounded for all $x \in \mathbb{X}$ and for all $t \in \{0, 1\}$.
\end{assumption}

\begin{assumption}
The joint density of the running vector, $f_X(\cdot)$, is twice continuously differentiable and bounded within some neighborhood of $x \in \partial \Omega$.
\end{assumption}

\begin{assumption}
The joint density of the running vector, $f_X(\cdot)$, is twice continuously differentiable and bounded above and away from zero within some neighborhood of $x \in \partial \Omega$, except where $x \in \partial \Omega$.
\end{assumption}

\section{Identification}
This section develops the identification framework for two proposed global test statistics, covering both conditional expectation discontinuities (treatment effect heterogeneity) and density discontinuities (manipulation).

\subsection{Treatment Effect Heterogeneity}
Define $\tau(x) = \mathbb{E}[Y_i(1) - Y_i(0) | X_i = x]$. Here, the hypothesis of interest is as follows:
\begin{flalign*}
H_0 &: \tau(x) = 0 \text{ for all } x \in \partial \Omega  \\
H_a &: \tau(x) \neq 0 \text{ for some }x \in \partial \Omega 
\end{flalign*}
% Add extensions in another version of this paper. 
    % Testing for constant treatment effects at some level
    % Addressing running variable scaling issue
    % Treatment effect heterogeneity by some covariate(s)

To test the hypothesis above without uniform confidence bands, the running vector $X_i$ is first aggregated into a scalar $g_i$. Using $g_i$ as a running variable leads to the following: 

\begin{theorem}
Under Assumptions 1-3, 5, and 7, we have:
\[
\lim_{\epsilon \rightarrow 0^+} \mathbb{E}[Y_i \mid g_i = \epsilon] - \lim_{\epsilon \rightarrow 0^-} \mathbb{E}[Y_i \mid g_i = \epsilon] 
= \frac{1}{f_{g_\Omega}(0)} \int_{x : g_\Omega(x) = 0} \tau(x) f_X(x) \, dS_\Omega,
\]
where \( dS_\Omega \) denotes a Hausdorff measure over $\partial \Omega$.
\end{theorem}

Proofs of Theorem 1 and all subsequent theorems are deferred to the appendix. Intuitively, Theorem 1 implies that when the signed distance function is used as a running variable, the standard RD difference identifies a weighted average across $\tau(x)$ for all $x \in \partial \Omega$. Although this quantity may be of interest to researchers, it is not sufficient to test the main hypothesis because $\mathbb{E}[Y_i(1) - Y_i(0) | g_i = 0]$ may equal zero even when there exists points $x \in \partial \Omega$ where $\tau(x) \neq 0$. This can happen when the positive discontinuities mix with negative discontinuities to produce a null result. To address this issue, let $\gamma(X_i) \equiv \gamma_i$ denote the closest boundary point to $X_i$ and let $\Gamma(X_i) \equiv \Gamma_i = 1$ if $\tau(\gamma(X_i)) \geq 0$ and $\Gamma_i = 0$ otherwise. This leads to the following result:

\begin{theorem}
Under Assumptions 1-3, 5, and 7, we have:
\begin{align*}
\tau_\Gamma \equiv \left(\lim_{\epsilon \rightarrow 0^+} \mathbb{E}[Y_i \mid g_i = \epsilon, \Gamma_i = 1] - \lim_{\epsilon \rightarrow 0^-} \mathbb{E}[Y_i \mid g_i = \epsilon, \Gamma_i = 1] \right) \mathbb{E}[\Gamma_i | g_i = 0]  -\\ \left(\lim_{\epsilon \rightarrow 0^+} \mathbb{E}[Y_i \mid g_i = \epsilon, \Gamma_i = 0] - \lim_{\epsilon \rightarrow 0^-} \mathbb{E}[Y_i \mid g_i = \epsilon, \Gamma_i = 0] \right)(1-\mathbb{E}[\Gamma_i | g_i = 0])
\\= \frac{1}{f_{g_\Omega}(0)} \int_{x : g_\Omega(x) = 0} |\tau(x)| f_X(x) \, dS_\Omega = \frac{1}{f_{g_\Omega}(0)} \int_{x : g_\Omega(x) = 0} (2\Gamma(x) - 1)\tau(x) f_X(x) \, dS_\Omega
\end{align*}
\end{theorem}

Here, $\tau_\Gamma$ represents the weighed average of $|\tau(x)|$ over $x \in \partial \Omega$. Intuitively, $\Gamma_i$ serves as a sign‐normalization factor: it reverses the sign of $Y_i$ whenever the corresponding treatment effect is negative, which aligns all effects in the same direction. In other words, the original hypothesis can be rewritten as:
\begin{flalign*}
H_0 &: \tau_\Gamma = 0   \\
H_a &: \tau_\Gamma > 0 
\end{flalign*}

To make estimation easier and more efficient, notice that $\tau_\Gamma$ can be rewritten as:
$$
\tau_\Gamma = \lim_{\epsilon \rightarrow 0^+} \mathbb{E}[(2\Gamma_i - 1)Y_i | g_i = \epsilon] - \lim_{\epsilon \rightarrow 0^-} \mathbb{E}[(2\Gamma_i - 1)Y_i | g_i = \epsilon]
$$
A key advantage of this approach is that it only uses the sign of $\tau(\gamma(X_i))$. As a result, $\Gamma_i$ can be estimated reliably even if the underlying $\tau(\gamma(X_i))$ is noisy, which is an important property in small samples or in settings with many running variables. In contrast, procedures that aggregate raw treatment effects along the boundary are fragile: a single poorly estimated effect can materially distort the resulting test statistic. This phenomenon is evident in Table 1. 

\subsection{Density Manipulation}
Consider the problem of testing whether the joint density of the running variables is continuous at the boundary. Such a test serves as a diagnostic for endogenous manipulation of the assignment variables. In this sense, the density-discontinuity test of \textcite{mccrary2008manipulation} is being extended to a multivariate setting. Define the following:
\begin{align*}
f_X^+(x) & = \lim_{x' \rightarrow x, g_\Omega(x') \geq 0} f_X(x') , \quad f_X^-(x) = \lim_{x' \rightarrow x, g_\Omega(x') < 0} f_X(x') \\
\tau_f(x) & = f_X^+(x) - f_X^-(x)
\end{align*}
In this case, the hypothesis of interest is as follows:
\begin{flalign*}
H_0 &: \tau_f(x) = 0 \text{ for all } x \in \partial \Omega  \\
H_a &: \tau_f(x) \neq 0 \text{ for some }x \in \partial \Omega 
\end{flalign*}
Using a change-of-variables for densities:
$$
f_{g_\Omega}(g) = \int_{x: g_\Omega(X) = g} f_X(x) dS_\Omega \implies \lim_{\epsilon \rightarrow 0^+} f_{g_\Omega}(\epsilon) - \lim_{\epsilon \rightarrow 0^-} f_{g_\Omega}(\epsilon) = \int_{x: g_\Omega(x) = 0} \tau_f(x)dS_\Omega 
$$
Let $\Lambda(X_i) \equiv \Lambda_i = 1$ if $\tau_f(\gamma(X_i)) \geq 0$ and $\Lambda_i = 0$ otherwise. With that, the result below follows:

\begin{theorem}
Under Assumption 1-3:
\begin{align*}
& \tau_\Lambda^f \equiv \left[ \lim_{\epsilon \rightarrow 0^+} f_{g_\Omega| \Lambda}(g = \epsilon | \Lambda_i = 1) - \lim_{\epsilon \rightarrow 0^-} f_{g_\Omega | \Lambda}(g = \epsilon | \Lambda_i = 1) \right] \mathbb{P}(\Lambda_i = 1) - \\
& \left[ \lim_{\epsilon \rightarrow 0^+} f_{g_\Omega | \Lambda}(g = \epsilon | \Lambda_i = 0) - \lim_{\epsilon \rightarrow 0^-} f_{g_\Omega | \Lambda}(g = \epsilon | \Lambda_i = 0) \right] \mathbb{P}(\Lambda_i = 0) \\
& = \int_{X: g_\Omega(X) = 0} |\tau_f(x)| dS_\Omega = \int_{X: g_\Omega(X) = 0} (2\Lambda(X) - 1)\tau_f(x) dS_\Omega 
\end{align*}
\end{theorem}

Here, $\tau_\Lambda^f$ is equivalent to the sum of $|\tau_f(x)|$ over all $x \in \partial \Omega$. Like in the boundary heterogeneity case, a multivariate density test can be written through the following hypothesis:
\begin{flalign*}
H_0 &: \tau_\Lambda^f = 0  \\
H_a &: \tau_\Lambda^f > 0
\end{flalign*}

To make this expression easier to estimate, $\tau_\Lambda^f$ can be rewritten as follows:
$$
\tau_\Lambda^f = \lim_{\epsilon \rightarrow 0^+} f_{(2\Lambda-1)g}(\epsilon) - \lim_{\epsilon \rightarrow 0^-} f_{(2\Lambda-1)g}(\epsilon)
$$

In other words, $\tau_\Lambda^f$ can be computed as the discontinuity in the density of $g_i^* \equiv (2\Lambda_i - 1)g_i$ at $g_i^* = 0$. 

\section{Estimation and Inference}
The following subsections describe estimation and inference procedures for both the treatment effect heterogeneity test and the density manipulation test. Although the two procedures are closely related, there are several important distinctions between them.

\subsection{Treatment Effect Heterogeneity}

Based on the identification results of the previous section, the quantity of interest is the following:
$$
\tau_\Gamma = \lim_{\epsilon \rightarrow 0^+} \mathbb{E}[Y_i^* | g_i = \epsilon] - \lim_{\epsilon \rightarrow 0^-} \mathbb{E}[Y_i^* | g_i = \epsilon]
$$
where $Y_i^* = (2\Gamma(X_i) - 1)Y_i$. However, $\Gamma(X_i)$ is not observed, so it will need to be estimated as $\Gamma_n(X_i)$. Now define $\hat{Y}_i = (2\Gamma_n(X_i) - 1)Y_i$. With that, notice the following:
\begin{align*}
\hat{Y}_i - Y_i^* & = Y_i(2\Gamma_n(X_i) - 1 - (2\Gamma(X_i) - 1)) \\
& = 2Y_i(\Gamma_n(X_i) - \Gamma(X_i)) \\
& \equiv \hat{\eta}_i
\end{align*}
In other words, $\hat{\eta}_i$ can be thought of as a measurement error term, where $\Gamma_i$ is a higher dimensional nuisance parameter. As in \textcite{chernozhukov2018double}, estimating $\Gamma_i$ and $\tau_\Gamma$ using the same data will cause problems for estimation and inference. To address these issues, $\hat{Y}_i$ will be generated as follows:
\begin{enumerate}
    \item[1.] Split the $n$ observations into $K$ folds.

    \item[2.] Denote all observations in fold $k \in \{1, 2, ..., K\}$ as $I_k$. For each fold $k$, use data from all other folds to estimate $\Gamma(\cdot)$. Call this estimate $\Gamma_{n,-k}(\cdot)$. 

    \item[3.] Define $\hat{Y}_{i,k} = (2\Gamma_{n,-k}(X_{i,k})-1)Y_{i,k}$ 
   
\end{enumerate}
where the index $k$ denotes the fold of observation $i$. To estimate $\Gamma(\cdot)$, the Local Linear Forest of \textcite{friedberg2020local} will be used. The advantage of this method is that it easily allows for efficient boundary point estimation with any number of running variables. 
\begin{lemma}
Let $\mu_n^{(llf)}(x_0)$ represent a local linear forest estimator of $\mu(x_0) \equiv \mathbb{E}[Y_i(t) | X_i = x_0]$ at boundary point $x_0$. Under assumptions 1-4 and 7-8, the local linear forest estimator is consistent for $\mu(x_0)$.
\end{lemma}
Once the outcome variable is defined, the fold-level version of $\tau_\Gamma$ (defined as $\tau_{\Gamma_{k,p}}$) will be estimated as follows:\footnote{The superscript $\pm$ is used as shorthand, with $+$ indexing quantities computed for observations with $g_{i,k} \ge 0$ and $-$ indexing quantities computed for observations with $g_{i,k} < 0$.}
\begin{align*}
\hat{\tau}_{\Gamma_{k, p}} & = e_1'(\hat{\beta}_{k,p}^+ - \hat{\beta}_{k,p}^-) \\
\hat{\beta}_{k,p}^\pm & = \argmin_{\beta \in \mathbb{R}^{p+1}} \sum_{i\in I_k} \delta_{i,k}^\pm \left( \hat{Y}_{i,k} - r_p(g_{i,k})\beta \right)^2 K\left( \frac{g_{i,k}}{h_k} \right) \\
\delta_{i,k}^+ & = \mathbbm{1}(g_{i,k} \geq 0), \delta_{i,k}^- = \mathbbm{1}(g_{i,k} < 0)
\end{align*}
where $K(\cdot)$ is some bounded, symmetric kernel function, $r_p(\cdot)$ is a vector of polynomial terms up to order $p$, and $h_k$ is the bandwidth. In this case, $K(\cdot)$ will be an Epanechnikov kernel and $p=1$. Additionally, the multiplier bootstrap approach of \textcite{chiang2019robust} will be used for inference and bandwidth selection, with polynomial order $q=2$ used for bias correction. More specifically, bootstrap estimates will be created as follows:
$$
\phi_{b,q}^\pm = \sum_{i=1}^n u_{i,b} \frac{e_1'(\Gamma_q^\pm)^{-1} (\hat{Y}_{i,k} - r_q(g_{i,k})\hat{\beta}_{k,q}^\pm)r_q\left( \frac{g_{i,k}}{h_k} \right) K\left( \frac{g_{i,k}}{h_k} \right) \delta_{i,k}^\pm}{n h_k \hat{f}_g(0)}
$$
where $u_{i,b}$ are independent and take values in $\{-1, 1\}$ (with equal probability), and $\Gamma_p^\pm$ is defined as follows:
$$
\Gamma_p^+ = \int_0^1 r_p(u) r_p(u)' K(u) du, \quad \Gamma_p^- = \int_{-1}^0 r_p(u) r_p(u)' K(u) du
$$

The variance of $\hat{\tau}_{\Gamma_{k, q}}$ can be computed using the variance of $\phi_{b,q}^+ - \phi_{b,q}^-$. With these components, the result below follows:\footnote{The theorem below can be trivially modified to accommodate two different bandwidths on either side of the cutoff.}

\begin{theorem}
Under assumptions 1-8, the following holds:
$$
\frac{1}{K} \sum_{k=1}^K \hat{\tau}_{\Gamma_{k, p}} \xrightarrow{p} \tau_\Gamma, \quad \frac{ \frac{1}{K} \sum_{k=1}^K \hat{\tau}_{\Gamma_{k, q}}- \tau_\Gamma}{\sqrt{\frac{1}{B-1} \sum_{b=1}^B (\phi_{b,q}^+ - \phi_{b,q}^- - (\bar{\phi}_q^+ - \bar{\phi}_q^-))^2}} \xrightarrow{d} \mathcal{N}(0, 1)
$$
where 
$$
\bar{\phi}_q^\pm = \frac{1}{B} \sum_{b=1}^b \phi_{b,q}^\pm
$$
\end{theorem}
An important observation is that asymptotic normality holds even though the local linear forest estimator converges more slowly than the second-stage local polynomial regression. The key reason is that, when a true discontinuity exists, inference depends only on correctly determining the signs of the treatment effects. Because the data lie on a bounded support, the probability of sign errors can be made arbitrarily small, regardless of the slower convergence rate of the first-stage estimator. As a result, the asymptotic properties of the standard local polynomial regression remain valid.

\subsection{Density Test}
Estimation and inference for the aggregated univariate density are carried out using a bias-corrected kernel density estimator, with sampling variability assessed via a multiplier bootstrap. A standard kernel density estimator is given as follows:
$$
\hat{f}_{g^*}(0)^+ - \hat{f}_{g^*}(0)^- =  \sum_{i=1}^n \left( \frac{1}{nh_k} \right) K\left( \frac{\hat{g}_{i,k}}{h_k} \right) (2\delta_i-1)
$$
where $K(\cdot)$ is a symmetric and bounded kernel function, and $\delta_i = \mathbbm{1}(\hat{g}_{i,k} \geq 0)$. Since the estimator is known to exhibit boundary bias, an approach similar to \textcite{hazelton2009linear} is adopted, whereby the base kernel is adjusted to achieve unbiasedness at the boundary. The resulting estimator is given by:
$$
\hat{f}_{g^*}(0)^+ - \hat{f}_{g^*}(0)^- =  \sum_{i=1}^n \left( \frac{1}{nh_k} \right) \left(\alpha_1 + \alpha_2 \left| \frac{\hat{g}_{i,k}}{h_k} \right| \right) K\left( \frac{\hat{g}_{i,k}}{h_k} \right)  (2\delta_i-1)
$$
where $\alpha_1, \alpha_2$ solves the following system of equations:
$$
\begin{pmatrix} \int_0^1 K(u) du & \int_0^1 u K(u) du \\ \int_0^1 u K(u) du & \int_0^1 u^2 K(u) du   \end{pmatrix} \begin{pmatrix} \alpha_1 \\ \alpha_2 \end{pmatrix}  = \begin{pmatrix}  1 \\ 0  \end{pmatrix}
$$
With this construction of $\alpha_1$ and $\alpha_2$, the boundary bias inherent in the standard kernel density estimator is eliminated, as the modified kernel integrates to one on the relevant side of the cutoff. These coefficients also remove the first-order bias, yielding behavior analogous to that of a kernel density estimator evaluated at an interior point. The optimal bandwidth is selected using the following MSE-optimal expression:
$$
h_k = \left( \frac{V_k}{B_k^2} \right)^\frac{1}{5}
$$
where $V_k$ is the variance term and $B_k$ is the bias term. Their exact expressions are given as follows:
\begin{align*}
V_k & = \frac{1}{n} \left(  \hat{f}_g(0)^+ \int_0^1  (\alpha_1 + \alpha_2 u)^2K(u)^2 du + \hat{f}_g(0)^- \int_{-1}^0  (\alpha_1 - \alpha_2 u)^2K(u)^2 du \right) \\
B_k & = \left(\int_0^1 u^2 (\alpha_1 + \alpha_2 u)K(u) du \right) \hat{f}_g''(0)^+ - \left(\int_{-1}^0 u^2 (\alpha_1 - \alpha_2 u)K(u) du \right) \hat{f}_g''(0)^-
\end{align*}
Here, the density and second derivative of the density will be estimated via preliminary local polynomial regressions (\textcite{cattaneo2020simple, cattaneo2024local}). For consistency of this aggregated estimator, a variation of the random forest density estimator used to construct $\hat{g}_{i,k}$ must be consistent. When joint density estimation is conducted at boundary points, the volume of each cell generated by the partitioning algorithm of \textcite{wen2022random} must be adjusted, particularly when the support of the running variables is non-rectangular. The effective volume of a cell is obtained by multiplying its original volume by the proportion of the cell that lies within the support of the running variables. This proportion is approximated by uniformly sampling points within the cell and computing the fraction that fall inside the support. With that, the following lemma holds:
\begin{lemma}
Without loss of generality, let $f_X(x | x \in \Omega)$ be the density of the running vector $X$ within the region $\Omega$. Also, let $\hat{f}_X(x | x \in \Omega)$ be the random forest density estimator of $f_X(x | x \in \Omega)$. Under assumptions 1-3 and 9, the random forest density estimator is consistent for $f_X(x | x \in \Omega)$.
\end{lemma}
As with the treatment effect heterogeneity case, a K-fold cross-fitting procedure is used to separate the estimation of $\hat{g}_{i,k}$ and the final test statistic. Also, as with other kernel-based nonparametric estimators, the MSE-optimal bandwidth cannot be used directly for inference because it does not adequately control for bias. To address this issue, the kernel bias-correction procedure of \textcite{calonico2018effect} is employed. Specifically, inference is conducted using the following re-centered estimator:
$$
\frac{1}{n} \sum_{i=1}^n \frac{1}{h_k} \left[ \left(\alpha_1 + \alpha_2 \left| \frac{\hat{g}_{i,k}}{h_k} \right| \right) K\left( \frac{\hat{g}_{i,k}}{h_k} \right) - L^{(2)}\left( \frac{\hat{g}_{i,k}}{h_k} \right) \int_0^1 \frac{u^2 (\alpha_1 + \alpha_2 u)K(u)}{2!} du \right]  (2\delta_i-1)
$$
where
$$
L(u) = 30(1-|u|)^2 u^2  \mathbbm{1}(|u| \leq 1)
$$
Here, $L(u)$ is designed to estimate the curvature of the given density at the boundary point $g = 0$. With that, the following will be used to generate bootstrap replications to compute standard errors:
\begin{align*}
\phi_b^\pm &= \sum_{i=1}^n \frac{u_i^{(b)}}{n h_k} \bigg[\left(\alpha_1 + \alpha_2 \left| \frac{\hat{g}_{i,k}}{h_k} \right| \right) K\left( \frac{\hat{g}_{i,k}}{h_k} \right)  - h_k \hat{f}_g(0)^\pm  \\
& -  \left( L^{(2)}\left( \frac{\hat{g}_{i,k}}{h_k} \right) - h_k \hat{f}_g''(0)^\pm \right) \int_0^1 \frac{u^2 (\alpha_1 + \alpha_2 u)K(u)}{2!} du  \bigg](2\delta_i-1)
\end{align*}
where 
$$
\hat{f}_g(0)^\pm = \begin{cases}
    \hat{f}_g(0)^+, \quad \delta_i = 1 \\
    \hat{f}_g(0)^-, \quad \delta_i = 0
\end{cases}, \quad \hat{f}_g''(0)^\pm = \begin{cases}
    \hat{f}_g''(0)^+, \quad \delta_i = 1 \\
    \hat{f}_g''(0)^-, \quad \delta_i = 0
\end{cases}
$$

Now, the result below follows:\footnote{The theorem below can be trivially modified to accommodate two different bandwidths on either side of the cutoff.}
\begin{theorem}
Under assumptions 1-3, 5-7 and 9, the following result holds:
$$
\frac{1}{K} \sum_{k=1}^K \hat{\tau}_{\Lambda_k} \xrightarrow{p} \tau_\Lambda^f, \quad \frac{\frac{1}{K} \sum_{k=1}^K (\hat{\tau}_{\Lambda_k} - \mathcal{B}_k) - \tau_\Lambda^f}{\sqrt{\frac{1}{B-1} \sum_{b=1}^B (\phi_b^+ - \phi_b^- - (\bar{\phi}^+ - \bar{\phi}^-))^2}} \xrightarrow{d} \mathcal{N}(0, 1)
$$
where $\mathcal{B}_k$ is the estimated bias from fold $k$.
\end{theorem}
As in the treatment effect heterogeneity setting, the asymptotic results for the bias-corrected kernel density estimator continue to hold. This is because the probability of assigning an incorrect sign to the estimated discontinuity can be driven to zero as the sample size grows, ensuring that such errors do not affect the limiting distribution.

\subsection{Accounting for Random Split Selection}
Asymptotically, the specific split used in the $K$-fold cross-fitting procedure is inconsequential. In finite samples, however, different splits can lead to non-negligible variation in both point estimates and standard errors. To account for this variability, the $K$-fold cross-fitting procedure is repeated $S$ times, and the resulting main estimators are averaged across the $S$ splits. For inference, influence functions are computed for each split; multiplier weights are then applied to the average influence function, with the weights for each observation $i$ held fixed across splits. This construction parallels the clustered multiplier bootstrap extension of \textcite{chiang2019robust} and incorporates the randomness of the sample split into the standard error calculation.

\section{Simulations}
This section illustrates the implementation of the treatment effect heterogeneity test and the density manipulation tests using simulated data. For each case, two data-generating processes (DGPs) are considered: one that is discontinuous almost everywhere, and one that is continuous everywhere.

\subsection{Treatment Effect Heterogeneity Simulations}
The first DGP will be constructed as follows:
$$
Y = \frac{1}{3} \begin{cases}
X_1 + X_2, \quad X_1, X_2 \in [0, 1] \\
1 - X_1, \quad X_1 \in [0, 1], X_2 \in [-1, 0] \\
1 - X_2, \quad X_1 \in [-1, 0], X_2 \in [0, 1] \\
1, \quad X_1, X_2 \in [-1, 0]
\end{cases} + \epsilon
$$
where $\epsilon \sim N(0, \sigma_\epsilon^2)$, $X_1, X_2$ take values in $[-1, 1]$, and $f_X(x) = 0.25$ with $X_1 \indep X_2$. Here, $\tau_\Gamma = \frac{1}{6}$. Below is a plot of this function: 

\begin{figure}[H]
    \centering
    \includegraphics[scale=0.7]{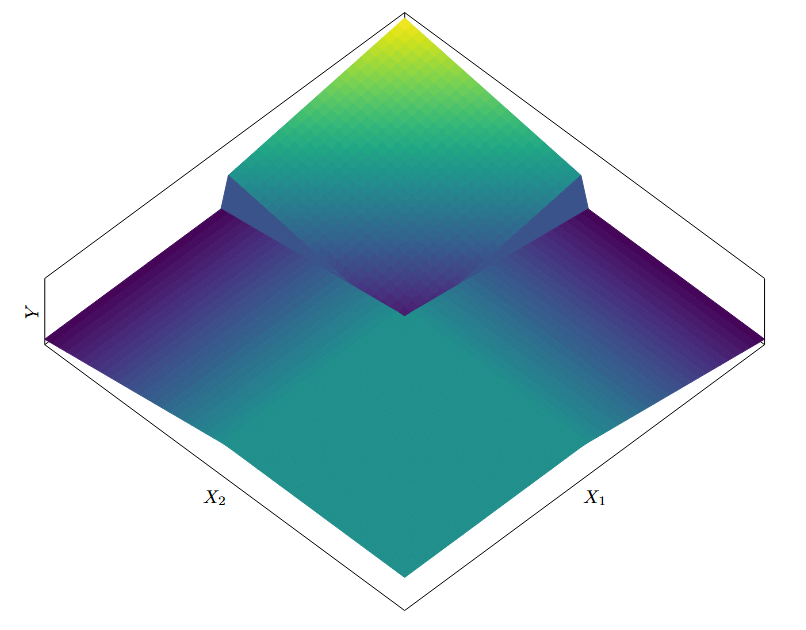}
    \caption{Graph of the DGP 1 function. Darker colors refer to lower values and lighter colors refer to higher values.}
\end{figure}

For the second DGP, the data will be generated as follows:
$$
Y = \frac{X_1 + X_2}{3} + \epsilon
$$
where $X_1$, $X_2$, and $\epsilon$ are constructed in the same way as in the first DGP. For each of these DGPs, the following parameters are used: sample size $n = 1000$, noise variance $\sigma_{\epsilon}^2 = 0.05$, number of folds $K = 2$, number of splits $S \in \{1,10,20\}$, and 5{,}000 simulation replications. Also, two different coverage-error optimal bandwidths are used on either side of the one-dimensional cutoff. The resulting simulation outcomes are presented below.

    \begin{table}[H]
    \centering
    \caption{Simulation results for treatment effect heterogeneity test using DGP 1 and DGP 2.}
    \begin{tabular}{lcc}
    \toprule
    & \textbf{DGP 1} & \textbf{DGP 2} \\
    \midrule
    \textbf{Bias} \\
    \quad Splits = 1 & -0.0295 & -0.0003 \\
    \quad Splits = 10 & -0.0295 & -0.0016 \\
    \quad Splits = 20 & -0.0294 & -0.0017 \\
    \addlinespace
    \textbf{Standard Errors} \\
    \quad Splits = 1 & 0.0895 & 0.0902 \\
    \quad Splits = 10 & 0.0743 & 0.0539 \\
    \quad Splits = 20 & 0.0734 & 0.0510 \\
    \addlinespace
    \textbf{Rejection Rate} \\
    \quad Splits = 1 & 0.5252 & 0.0778 \\
    \quad Splits = 10 & 0.6412 & 0.0774 \\
    \quad Splits = 20 & 0.6518 & 0.0762 \\
    \bottomrule
    \textbf{Rejection Rate (Local Polynomial Regression from Table 1)} \\
    \quad $n = 1,000$ & & 0.3652 \\
    \quad $n = 20,000$ & & 0.0748 
    \end{tabular}
    \end{table}

These results reveal several patterns. First, the bias under the first DGP exceeds that under the second. This occurs because, in the second DGP, the sign of the outcome is immaterial in the absence of discontinuities. In contrast, under the first DGP, misclassification of the treatment effect sign has a non-trivial influence on the test statistic, leading to greater bias. Consequently, DGPs featuring discontinuities are expected to exhibit slightly larger bias. This bias does not seem to change on average as the number of splits increases. \\

Second, the standard errors exhibit substantial sensitivity to the number of splits. For both DGPs, the standard error is approximately 0.09 with a single split, but decreases markedly when the number of splits is increased to 10. Although the initial reduction is pronounced, further increases in the number of splits yield progressively smaller declines. For the discontinuous DGP, this reduction improves the rejection rate, whereas the rejection rate remains relatively stable for the continuous DGP.

\subsection{Density Manipulation Simulations}
For the first density test example, the following bivariate density will be used:
$$
f_X(X_1, X_2) = \frac{1}{3} \begin{cases}
X_1 + X_2, \quad X_1, X_2 \in [0, 1] \\
1 - X_1, \quad X_1 \in [0, 1], X_2 \in [-1, 0] \\
1 - X_2, \quad X_1 \in [-1, 0], X_2 \in [0, 1] \\
1, \quad X_1, X_2 \in [-1, 0]
\end{cases}
$$
Note that this is the same function that was used in DGP 1 for the boundary heterogeneity simulations. Here, $\tau_{\Lambda}^f = 1/3$. For the second density, the following joint density will be used:
$$
f_X(X_1, X_2) = 1/4
$$
where $X_1, X_2$ are independent and take values on $[-1, 1]$. Here, the first density is discontinuous almost everywhere along the treatment boundary, whereas the second density is continuous everywhere. For bandwidth selection, a single MSE optimal bandwidth is used on either side of the one-dimensional cutoff. For $n = 1000$, $K = 2$ and $S = 1,10,20$, notice the following simulation results: \\

\begin{table}[H]
\centering
\caption{Simulation results for density test using DGP 1 and DGP 2.}
\begin{tabular}{lcc}
\toprule
& \textbf{DGP 1} & \textbf{DGP 2} \\
\midrule
\textbf{Bias} \\
\quad Splits = 1 & -0.035 & -0.002 \\
\quad Splits = 10 & -0.036 & -0.001 \\
\quad Splits = 20 & -0.035 & -0.002 \\
\addlinespace
\textbf{Standard Errors} \\
\quad Splits = 1 & 0.135 & 0.167 \\
\quad Splits = 10 & 0.111 & 0.095 \\
\quad Splits = 20 & 0.109 & 0.090 \\
\addlinespace
\textbf{Rejection Rate} \\
\quad Splits = 1 & 0.721 & 0.059 \\
\quad Splits = 10 & 0.862 & 0.063 \\
\quad Splits = 20 & 0.874 & 0.061 \\
\bottomrule
\end{tabular}
\end{table}

In this scenario, the simulation results exhibit a pattern similar to that observed in the treatment effect heterogeneity case. Specifically, the bias is slightly greater under the discontinuous DGP, and the standard errors decrease up to a certain point as the number of splits increases.

\section{Application}
The empirical application in this paper follows the setting of \textcite{frey2019cash}. In 2003, the conditional cash transfer (CCT) program \textit{Bolsa Família} (BF) was launched across Brazil. Beyond serving as a major source of income stabilization for extremely poor households, there is substantive speculation that BF also reduced the scope for political influence by diminishing incentives to engage in clientelism.\footnote{Here, clientelism refers to the act of "...replacing public good distribution with private transfers targeting groups or individual voters" (\textcite{frey2019cash}).} To evaluate this claim, \textcite{frey2019cash} exploit cross-municipal variation in CCT coverage induced by the Family Health Program (FHP). Beginning in August 2004, municipalities with fewer than 30,000 residents and a Human Development Index (HDI) below 0.7 became eligible for a 50\% increase in FHP funding. This expansion influenced BF coverage by improving households’ access to information about the program. According to a survey of 10,000 poor households, more than 10\% of BF beneficiaries learned about the program through their family doctor. Consequently, the political effects of BF can be identified using variation in BF coverage generated by the FHP expansion through an instrumental variables or regression discontinuity analysis. \\

In the original \textcite{frey2019cash} study, the author implements a multivariate regression discontinuity design along the FHP funding threshold to assess the strength of the first stage and to estimate the effects of BF on various political outcomes. However, with only a few thousand municipalities, a full boundary-based treatment effect analysis may lack reliability due to limited effective sample size near each cutoff point. To evaluate the robustness of the original findings, the methods of this paper will be applied to test for the presence of any treatment effects along the FHP boundary for the same outcomes examined in \textcite{frey2019cash}. For simplicity, only municipalities from 2008 are included in this analysis. \\

For the first set of results, the outcomes will be health funds received and CCT coverage. These outcomes measure the extent to which the intervention affects the type of treatment that was received. The first column reports the weighted average absolute treatment effect, with p-values in brackets.  The second column presents the weighted average treatment effect based on simple signed distances (as in Theorem 1). The third column indicates whether a discontinuity is detected using the uniform bands of \textcite{cattaneo2025estimation}. The fourth column reports whether both positive and negative discontinuities are detected using the local polynomial regression approach of \textcite{cattaneo2025estimation}. All running variables are standardized using the standard deviation.

\begin{table}[H]
\centering
\caption{Treatment effect heterogeneity test for first-stage outcomes.}
\label{tab:rdd_results}
\begin{tabular}{lcccc}
\toprule
\textbf{Variable} & \textbf{Est.} & \textbf{Est. (Dist.)} & \textbf{LPR Disc.} & \textbf{Double Disc.} \\
\midrule
Health Funds, R\$mn  & 0.352 & 0.196 & No & No \\
                     & [0.714] & [0.436] &  &  \\
CCT Coverage, p.p.   & 3.217 & 4.036 & No & No \\
                     & [0.426] & [0.373] &  &  \\
\bottomrule
\end{tabular}

\vspace{0.5em}
\footnotesize\textit{Note:} P-values are reported in square brackets.
\end{table}

Unlike \textcite{frey2019cash}, the analysis presented here does not uncover evidence of a strong first stage. This conclusion holds both when applying the uniform confidence bands of \textcite{cattaneo2025estimation} and when using the aggregation methods developed in this paper. Together, these results indicate that the available data do not provide sufficient support for a strong first-stage effect. At the same time, there may be scope to improve the power of these tests without compromising size. As noted by \textcite{crippa2025manipulation, wong2013analyzing}, the scaling of the running variables can meaningfully affect the performance of distance-based aggregators. Intuitively, rescaling alters the relative emphasis placed on different segments of the boundary. A common remedy is to standardize each running variable so that all coordinates receive equal weight. However, there is no reason to expect equal weighting to be optimal for power. For instance, if discontinuities occur primarily along one portion of the boundary, assigning greater weight to that region may deliver more powerful tests than treating all boundary segments symmetrically. One possible way to address this issue is to treat the scaling parameters as tuning choices and select them to maximize power. Although this extension lies beyond the scope of this paper, it may improve the power of the proposed test and potentially reveal first-stage discontinuities that remain undetected under equal weighting. \\

For the second set of results, the analysis will use political outcomes that will help determine whether BF had any impact on clientelism in Brazil. 

\begin{table}[H]
\centering
\caption{Treatment effect heterogeneity test for main political outcomes.}
\label{tab:rdd_main_results}
\begin{tabular}{l c c c c}
\toprule
\textbf{Variable} & \textbf{Est.}  & \textbf{Est. (Dist.)} & \textbf{LPR Disc.} & \textbf{Double Disc.} \\
\midrule
Incumbent's Vote Share (\%) & \makecell{4.408 \\ \text{[0.135]}} &  \makecell{-2.163 \\ \text{[0.528]}} & No & No \\
Margin of Victory (p.p.) & \makecell{0.039 \\ \text{[0.288]}} & \makecell{-4.650 \\ \text{[0.248]}} & No & No \\
Candidates (Number) & \makecell{-0.059 \\ \text{[0.663]}} & \makecell{-0.056 \\ \text{[0.531]}} & No & No \\
Pro-Poor Spending (\% share) & \makecell{11.945** \\ \text{[0.015]}} & \makecell{2.643 \\ \text{[0.173]}} & Yes & Yes \\
Challenger's Entry (share with HS) & \makecell{-0.017 \\ \text{[0.653]}} & \makecell{-0.075 \\ \text{[0.305]}} & No & No \\
Challenger's Entry (share clientelistic) & \makecell{-0.065 \\ \text{[0.990]}} &  \makecell{0.092 \\ \text{[0.252]}} & No & No \\
Challenger is Top 2 (share with HS) & \makecell{-0.025 \\ \text{[0.530]}} &  \makecell{-0.155 \\ \text{[0.124]}} & No & No \\
Challenger is Top 2 (share clientelistic) & \makecell{-0.090 \\ \text{[0.996]}} & \makecell{0.124 \\ \text{[0.197]}} & No & No \\
\bottomrule
\end{tabular}

\vspace{0.5em}
\footnotesize\textit{Note:} P-values are reported in square brackets.
\end{table}

In this case, all but one of the political outcomes retain their significance from \textcite{frey2019cash}. In particular, pro-poor spending exhibits a clear discontinuity at the FHP boundary. Notably, both the local polynomial estimator and the ML-based distance aggregation approach detect this discontinuity, whereas the standard distance-based aggregation estimator does not. A closer examination of the local polynomial estimates suggests that pro-poor spending increases at some boundary locations and decreases at others. Below is a figure showing heterogeneity in the treatment effects estimated by the local polynomial regression. 

\begin{figure}[H]
    \centering
    \includegraphics[width=0.8\linewidth]{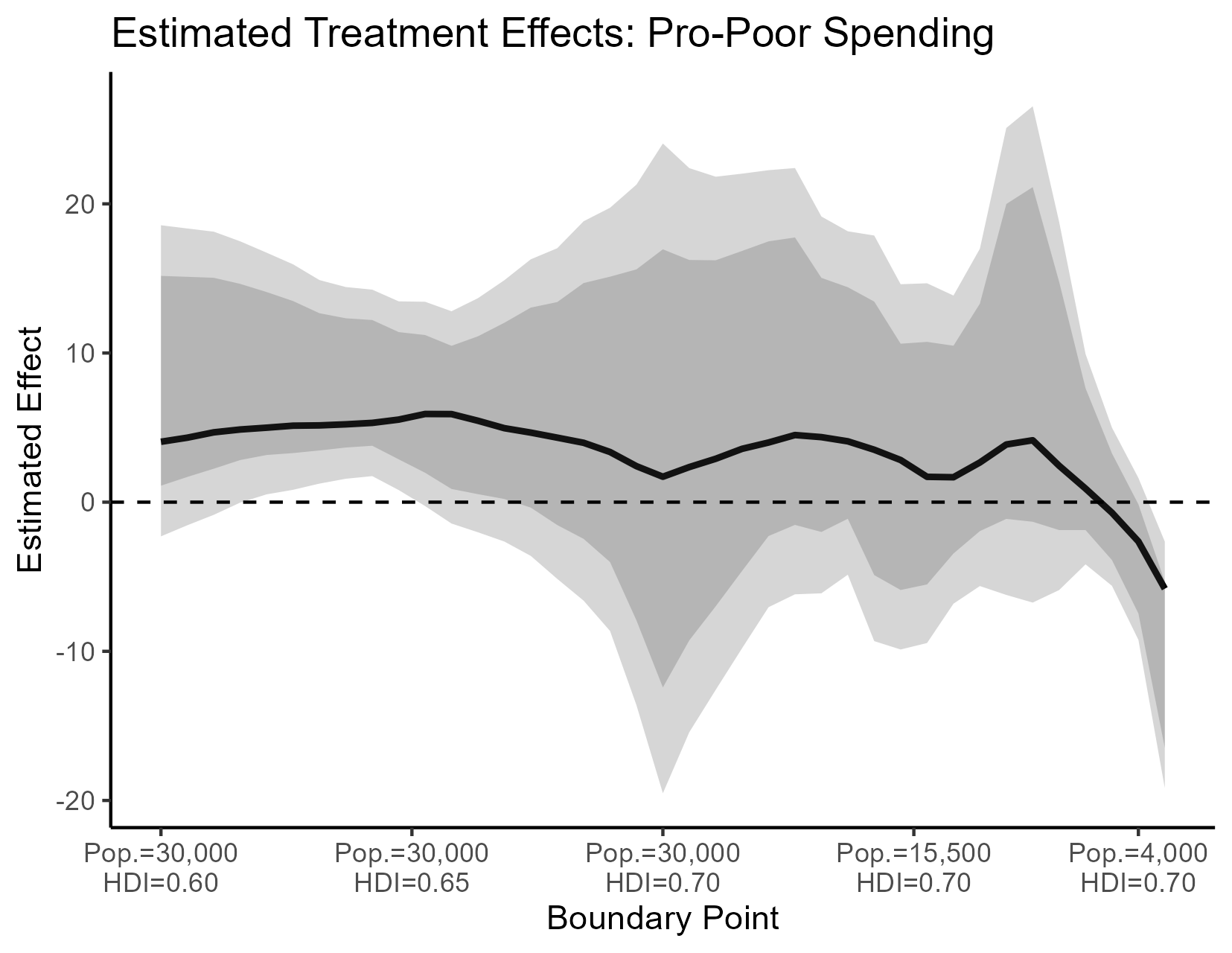}
    \caption{Treatment effects for pro-poor spending outcome using a local polynomial regression. The darker region represents pointwise bands and the lighter region represents uniform bands.}
\end{figure}

The patterns observed in pro-poor spending are intuitive: larger and less-developed municipalities tend to exhibit higher pro-poor expenditures than smaller and more-developed municipalities. Under the standard distance-based aggregation method, the positive and negative discontinuities offset one another, attenuating the aggregated statistic toward zero. In contrast, the aggregation method proposed in this paper avoids this cancellation issue and is therefore able to detect a discontinuity despite the heterogeneous pattern of effects along the boundary. \\

To check for robustness, a density test will be conducted to ensure that municipalities did not manipulate their running variables to receive more funding. More specifically, a density test is conducted for four different subsets of the data, each of which were used in the main analysis. For the first subset, only municipalities in which the candidate was the incumbent mayor were retained. The second subset includes municipalities where the incumbent mayor was eligible to run for re-election. The third subset restricts the sample to municipalities in which the candidate was not the incumbent but the incumbent mayor was running for re-election. Finally, the fourth subset applies the same criteria as the third, further restricting the sample to candidates whose rank was less than three. These tests show no signs of manipulation.
\begin{table}[H]
\centering
\caption{Density test results for different running variable subsets.}
\label{tab:rdd_results_density}
\begin{tabular}{lcccc}
\toprule
\textbf{Variable} & \textbf{Est.}  & \textbf{Est. (Dist.)}   \\
\midrule
Candidate is Incumbent Mayor  & -0.046 & -0.045 \\
& [0.369]  & [0.591] &  \\
Incumbent Mayor Allowed for Re-election &  -0.258  & -0.303  \\
& [0.428] & [0.177] & \\
Not Incumbent Mayor, Running for Re-election  & 0.045   & -0.062  \\
& [0.330] & [0.327] & \\
Not Incumbent Mayor, Running for Re-election, Rank less than 3  &  -0.034  & -0.066  \\
& [0.339] & [0.639] & \\
\bottomrule
\end{tabular}

\vspace{0.5em}
\footnotesize\textit{Note:} P-values are reported in square brackets.
\end{table}

\section{Conclusion}
This paper develops a testing procedure for detecting discontinuities along multidimensional regression discontinuity (RD) boundaries. A central advantage of the proposed approach is its strong performance in relatively small samples, which is a setting where many existing multivariate RD methods struggle due to the curse of dimensionality. The simulation evidence demonstrates that the procedure yields reliable inference even with limited data, and the empirical application illustrates how it can complement and strengthen insights obtained from multivariate RD estimators. \\

Several promising directions for future research remain. One important extension concerns the choice of distance metric used to aggregate local estimates along the boundary. Because the power of the global test depends on how distances are scaled, developing principled data-driven methods for tuning or learning the metric may lead to substantial power gains. Another avenue involves identifying where along the boundary discontinuities occur. Combining the proposed global test with localization techniques could provide a more refined understanding of heterogeneity in treatment effects or sorting, while limiting the disadvantages of standard multivariate RD estimators.

\pagebreak
\printbibliography

\pagebreak

\section{Appendix}
Here, proofs and materials for the main paper are included.

\subsection{Consistency of Random Forest Estimators}

\subsubsection{Proof of Lemma 1}
\begin{proof}
Let $x_0$ be an evaluation point and, without loss of generality, consider estimation from the region $\Omega$. The standard random forest estimator can be written as follows:
$$
\mu_n^{(rf)}(x_0)^+ = \sum_{i=1}^n \alpha_i^+(x_0) Y_i
$$
where $\alpha_i^+(x_0) = 0$ if $X_i \notin \Omega$. Let $\boldsymbol{X}$ be an $n \times d+1$ matrix with $\begin{pmatrix} 1 & (X_i - x_0)' \end{pmatrix}$ in each row, let $A^+ \in \mathbb{R}^{n \times n}$ be a diagonal matrix with $\alpha_i^+(x_0)$ on the diagonal, and let $Y$ be a vector of size $n$ consisting of each $Y_i$. Also, let $J = \text{diag}(0, 1, 1, 1, ...)$ and $\lambda_n \to 0$ (as $n \to \infty$) be the ridge regularization penalty. With that, the local linear forest estimator can be written as follows:
\begin{align*}
\mu_n^{(llf)}(x_0) & = e_1' (\boldsymbol{X}' A^+\boldsymbol{X} + \lambda_n J)^{-1} \boldsymbol{X}' A^+ Y \\
& =  e_1' (\boldsymbol{X}' A^+\boldsymbol{X} + \lambda_n J)^{-1} \sum_{i=1}^n \begin{pmatrix} 1 \\ X_i - x_0 \end{pmatrix}\alpha_i^+(x_0) Y_i
\end{align*}
With that, this estimator can be simplified as follows: 
\begin{align*}
\mu_n^{(llf)}(x_0) & = e_1' (\boldsymbol{X}' A^+\boldsymbol{X} + \lambda_n J)^{-1} \sum_{i=1}^n \begin{pmatrix} 1 \\ X_i - x_0 \end{pmatrix}\alpha_i^+(x_0) Y_i \\
& = e_1' (\boldsymbol{X}' A^+\boldsymbol{X} + \lambda_n J)^{-1} \sum_{i=1}^n \begin{pmatrix} 1 \\ X_i - x_0 \end{pmatrix}\alpha_i^+(x_0) (\mu(X_i) + \epsilon_i)
\end{align*}
If $\mathbbm{1}$ is a vector of all ones, then a Taylor expansion of $\mu(\cdot)$ gives the following:
\begin{align*}
\mu_n^{(llf)}(x_0)^+ & =  e_1' (\boldsymbol{X}' A^+\boldsymbol{X} + \lambda_n J)^{-1} \sum_{i=1}^n \begin{pmatrix} 1 \\ X_i - x_0 \end{pmatrix}\alpha_i^+(x_0) \epsilon_i \\
& + e_1' (\boldsymbol{X}' A^+\boldsymbol{X} + \lambda_n J)^{-1} \sum_{i=1}^n \begin{pmatrix} 1 \\ X_i - x_0 \end{pmatrix}\alpha_i^+(x_0) \mu(x_0)^+ \\
& + e_1' (\boldsymbol{X}' A^+\boldsymbol{X} + \lambda_n J)^{-1} \sum_{i=1}^n \begin{pmatrix} 1 \\ X_i - x_0 \end{pmatrix}\alpha_i^+(x_0) \nabla_x\mu(x_0)'^+ (X_i - x_0) \\
& + O(\bar{R}^2)
\end{align*}
where $\bar{R}^2$ is the average squared radius of the leaves in each tree grown by the random forest. Now, since $\boldsymbol{X}e_1 = \mathbbm{1}$ and $\lambda_n J e_l = 0$, notice the following:
\begin{align*}
(\boldsymbol{X}' A^+ \boldsymbol{X} + \lambda_n J) e_1 & = \boldsymbol{X}' A^+ \mathbbm{1} \\
e_1 & = (\boldsymbol{X}' A^+ \boldsymbol{X} + \lambda_n J)^{-1} \boldsymbol{X}' A^+ \mathbbm{1} \\
e_1' e_1 & = e_1' (\boldsymbol{X}' A^+ \boldsymbol{X} + \lambda_n J)^{-1} \boldsymbol{X}' A^+ \mathbbm{1} \\
& = e_1' (\boldsymbol{X}' A^+ \boldsymbol{X} + \lambda_n J)^{-1} \sum_{i=1}^n \begin{pmatrix} 1 \\ X_i - x_0 \end{pmatrix}\alpha_i^+(x_0)  \\
& = 1
\end{align*}
With that, the estimator can be simplified as follows:
\begin{align*}
\mu_n^{(llf)}(x_0)^+ & = \mu(x_0)^+ \\ & +  e_1' (\boldsymbol{X}' A^+\boldsymbol{X} + \lambda_n J)^{-1} \sum_{i=1}^n \begin{pmatrix} 1 \\ X_i - x_0 \end{pmatrix}\alpha_i^+(x_0) \epsilon_i \\
& + e_1' (\boldsymbol{X}' A^+\boldsymbol{X} + \lambda_n J)^{-1} \sum_{i=1}^n \begin{pmatrix} 1 \\ X_i - x_0 \end{pmatrix}\alpha_i^+(x_0) \nabla_x\mu(x_0)' (X_i - x_0) \\
& + O(\bar{R}^2)
\end{align*}
From here, since $\lambda_n \to 0$, notice the following about the third term.
\begin{align*}
& e_1' (\boldsymbol{X}' A^+\boldsymbol{X} + \lambda_n J)^{-1} \sum_{i=1}^n \begin{pmatrix} 1 \\ X_i - x_0 \end{pmatrix}\alpha_i^+(x_0) \nabla_x\mu(x_0)' (X_i - x_0) \\
& = e_1' (\boldsymbol{X}' A^+\boldsymbol{X} + \lambda_n J)^{-1} \sum_{i=1}^n \begin{pmatrix} 1 \\ X_i - x_0 \end{pmatrix}\alpha_i^+(x_0) \begin{pmatrix} 0 \\ \nabla_x \mu(x_0) \end{pmatrix}' \begin{pmatrix} 1 \\ X_i - x_0 \end{pmatrix} \\
& = e_1' (\boldsymbol{X}' A^+\boldsymbol{X} + \lambda_n J)^{-1} \sum_{i=1}^n \begin{pmatrix} 1 \\ X_i - x_0 \end{pmatrix}\alpha_i^+(x_0) \begin{pmatrix} 1 \\ X_i - x_0 \end{pmatrix} ' \begin{pmatrix} 0 \\ \nabla_x \mu(x_0) \end{pmatrix} \\
& = e_1' (\boldsymbol{X}' A^+\boldsymbol{X} + \lambda_n J)^{-1} \boldsymbol{X}' A^+\boldsymbol{X} \begin{pmatrix} 0 \\ \nabla_x \mu(x_0) \end{pmatrix} \\
& \to 0
\end{align*}
With that, the estimator can be further simplified as follows.
\begin{align*}
\mu_n^{(llf)}(x_0)^+ & = \mu(x_0)^+ \\ 
& +  e_1' (\boldsymbol{X}' A^+\boldsymbol{X} + \lambda_n J)^{-1} \sum_{i=1}^n \begin{pmatrix} 1 \\ X_i - x_0 \end{pmatrix}\alpha_i^+(x_0) \epsilon_i \\
& + O(\bar{R}^2)
\end{align*}
Since $\epsilon_i$ is independent of $\alpha_i(x_0)$ and $X_i$, the expectation of the second term above is equal to zero. Therefore, we have:
$$
 e_1' (\boldsymbol{X}' A^+\boldsymbol{X} + \lambda_n J)^{-1} \sum_{i=1}^n \begin{pmatrix} 1 \\ X_i - x_0 \end{pmatrix}\alpha_i^+(x_0) \epsilon_i \xrightarrow{p} 0
$$
Now, note that the radius of each leaf vanishes according to Lemma 1 of \textcite{masini2025balancing}. Therefore, the local linear forest estimator is consistent at any boundary point $x_0$. \end{proof}

\subsubsection{Proof of Lemma 2}
\begin{proof}
The random forest density estimator is an approach that approximates the density through an averaging of randomly partitioned trees. This approach is similar to a histogram estimator, except where the partitions in the random forest density estimator are made using a more efficient algorithm. \\

Let $T$ be the number of trees and let $d$ be the number of running variables. Also, without loss of generality, let $n$ denote the number of points in $\Omega$ and let $N$ denote the total number of points. With that, let $p_n$ be the depths of the splits.\footnote{In this case, set $p_n = p_N = d(1+d\ln(2))^{-1} \ln(N)$, which appears in Theorem 2 of \textcite{wen2022random}.} The random forest density estimator comes from averaging over many random tree density estimators. To determine the splits used for a given tree, the first step is enclosing the support of $X_i$ into a hyper rectangle $B_r = [-r, r]^d$. Then, a coordinate $l \in \{1, 2, ..., d\}$ is randomly chosen along which the data is split in half. This process repeats for each individual partition of the data until depth $p_n$ is met. \\

Let $A_p^{j,t}$ denote leaf $j$ of tree $t$ with depth $p$ and let $I_p^t$ denote the index of of each partition for tree $t$ and depth $p$. Also, let $\mu(A_p^{j,t})$ be a non-zero Lebesgue measure of each leaf and let $D(A_p^{j,t})$ be defined as follows:
$$
D(A_p^{j,t}) = \frac{1}{n} \sum_{i \in \Omega} \mathbbm{1}(X_i \in A_p^{j,t})
$$
With that, the random tree density estimator is defined as follows:
$$
f_{D, t}^p(x_0 | x \in \Omega) = \sum_{j \in I_p^t} \frac{D(A_p^{j,t})\mathbbm{1}(x_0 \in A_p^{j,t}(x_0))}{\mu(A_p^{j,t})}
$$
The random forest density estimator is then defined as follows:
$$
f_{D, E}(x_0 | x \in \Omega) = \frac{1}{T} \sum_{t=1}^T f_{D, t}^p(x_0)
$$

If $D(\cdot)$ is the empirical probability measure, let $P(\cdot)$ denote the corresponding population probability measure. With that, notice the following:
\begin{align*}
\mathbb{E}[(f_{D, E}(x_0 | x \in \Omega) - f(x_0 | x \in \Omega))^2] & = \mathbb{E}[(f_{D, E}(x_0 | x \in \Omega) - f_{P, E}(x_0 | x \in \Omega))^2] \\
& + \mathbb{E}[(f_{P, E}(x_0 | x \in \Omega) - f(x_0 | x \in \Omega))^2] \\
& + 2 \mathbb{E}[(f_{D, E}(x_0 | x \in \Omega) - f(x_0 | x \in \Omega)) \\
& \cdot (f_{P, E}(x_0 | x \in \Omega) - f(x_0 | x \in \Omega))] \\
& \leq \mathbb{E}[(f_{D, E}(x_0 | x \in \Omega) - f_{P, E}(x_0 | x \in \Omega))^2] \\
&+ \mathbb{E}[(f_{P, E}(x_0 | x \in \Omega) - f(x_0 | x \in \Omega))^2] \\ 
& + 2 \sqrt{\mathbb{E}[(f_{D, E}(x_0 | x \in \Omega) - f_{P, E}(x_0 | x \in \Omega))^2]} \\
& \cdot \sqrt{\mathbb{E}[(f_{P, E}(x_0 | x \in \Omega) - f(x_0 | x \in \Omega))^2]}
\end{align*}
Therefore, if $\mathbb{E}[(f_{D, E}(x_0 | x \in \Omega) - f_{P, E}(x_0 | x \in \Omega))^2]$ and $\mathbb{E}[(f_{P, E}(x_0 | x \in \Omega) - f(x_0 | x \in \Omega))^2]$ converge to zero with the sample size, then the random forest density estimator is consistent for the true density. From the results of \textcite{wen2022random} for $f \in C^{1,\alpha}$:
\begin{align*}
\mathbb{E}[(f_{P, E}(x_0 | x \in \Omega) - f(x_0 | x \in \Omega))^2] & \leq \frac{c^2 (2r)^4 d^2}{T} \exp\left( \frac{-0.75p_n}{d} \right) + 4c^2(2r)^{2d+2} d^2 \exp\left( \frac{-p_n}{d} \right) 
\end{align*}
for some constant $c$.\footnote{Even with an irregular support, these bounds still work because the diameter of each cell is bounded above by the diameter given in the original \textcite{wen2022random} paper.} Since $p_n \to \infty$ as $n \to \infty$, $\mathbb{E}[(f_{P, E}(x_0 | x \in \Omega) - f(x_0 | x \in \Omega))^2] \to 0$. For the other term:
\begin{align*}
\mathbb{E}[(f_{D, E}(x_0 | x \in \Omega) - f_{P, E}(x_0 | x \in \Omega))^2] & \leq \mathbb{E} \left[ \frac{P(A_p^j(x_0))}{n \mu^2(A_p^j(x_0))} \right] \\
& \leq \mathbb{E} \left[ \frac{c_f}{n \mu(A_p^j(x_0))} \right]
\end{align*}
for some constant $c_f$. Note that when $\mu(A_p^j(x_0)) = 0$, $\mathbb{E}[(f_{D, E}(x_0 | x \in \Omega) - f_{P, E}(x_0 | x \in \Omega))^2] = 0$. Let $V_j \in [v_{\text{min}}, 1]$ be the fraction of each cube that is within the support of $X$. By construction, notice the following:
\begin{align*}
n \mu(A_p^j(x_0)) & = \frac{n(2r)^d}{2^p}V_j \\
\ln(n \mu(A_p^j(x_0)) & = d\ln(2r) + \ln(n) - p\ln(2) + \ln(V_j) \\
& = d\ln(2r) + \ln(n) - \frac{d\ln(n)\ln(2)}{1+d\ln(2)} + \ln(V_j) \\
& = d\ln(2r) + \ln(n) \left( 1 -  \frac{d\ln(2)}{1+d\ln(2)}\right)  + \ln(V_j) \\
& \to \infty
\end{align*}
In other words, $n \mu(A_p^j(x_0)) \to \infty$, which means that $\mathbb{E}[(f_{D, E}(x_0 | x \in \Omega) - f_{P, E}(x_0 | x \in \Omega))^2] \to 0$, thus showing that $f_{D, E}(x_0 | x \in \Omega)$ is consistent for $f(x_0 | x \in \Omega)$. Finally, to get the full density from the region $\Omega$, $f(x_0 | x \in \Omega)$ must be multiplied by $\frac{n}{N}$.
\end{proof}

\subsection{Proofs of Theorems}
\subsubsection{Proof of Theorem 1}
\begin{proof}
First, define $Y^{(d)} = Y(1) - Y(0)$. By Assumption 7, since $g_i$ is a deterministic function of $X_i$, $\mathbb{E}[Y_i(t) | g_i = g]$ is continuous for all $g$ and for all $t \in \{0, 1\}$. With that, notice the following:
\begin{align*}
\lim_{\epsilon \rightarrow 0^+} \mathbb{E}[Y_i \mid g_i = \epsilon] - \lim_{\epsilon \rightarrow 0^-} \mathbb{E}[Y_i \mid g_i = \epsilon] 
& = \lim_{\epsilon \rightarrow 0^+} \mathbb{E}[Y_i(1) \mid g_i = \epsilon] - \lim_{\epsilon \rightarrow 0^-} \mathbb{E}[Y_i(0) \mid g_i = \epsilon] \\
& = \mathbb{E}[Y_i(1) - Y_i(0) | g_i = 0] \\
& = \int_y Y^{(d)} f_{Y^{(d)}|g}(Y^{(d)}|g=0) dY^{(d)} \\
& = \frac{1}{f_{g_\Omega}(0)} \int_y Y^{(d)} f_{Y^{(d)}, g}(Y^{(d)}, g=0) dY^{(d)} \\
& = \frac{1}{f_{g_\Omega}(0)} \int_y \int_{X: g_\Omega(X) = 0} Y^{(d)} f_{Y^{(d)}, X}(Y, X) dS_\Omega dY^{(d)} \\
& = \frac{1}{f_{g_\Omega}(0)} \int_{X: g_\Omega(X) = 0} \left( \int_y  Y^{(d)} f_{Y^{(d)}|X}(Y^{(d)}|X) dY^{(d)} \right)  f_X(x) dS_\Omega \\
& = \frac{1}{f_{g_\Omega}(0)} \int_{x : g_\Omega(x) = 0} \tau(x) f_X(x) \, dS_\Omega
\end{align*}
\end{proof}

\subsubsection{Proof of Theorem 2}
\begin{proof}
First, note that the following holds by the definition of conditional density:
\begin{align*}
\frac{\mathbb{P}(\Gamma = 1 | g = 0)}{f_{g, \Gamma}(g=0, \Gamma = 1)} & = \frac{1}{f_{g_\Omega}(0)} \\
& = \frac{\mathbb{P}(\Gamma = 0 | g = 0)}{f_{g, \Gamma}(g=0, \Gamma = 0)}
\end{align*}
With that, notice the following:
\begin{align*}
& \left(\lim_{\epsilon \rightarrow 0^+} \mathbb{E}[Y_i \mid g_i = \epsilon, \Gamma_i = 1] - \lim_{\epsilon \rightarrow 0^-} \mathbb{E}[Y_i \mid g_i = \epsilon, \Gamma_i = 1] \right)\mathbb{E}[\Gamma_i | g_i = 0] -\\ & \left(\lim_{\epsilon \rightarrow 0^+} \mathbb{E}[Y_i \mid g_i = \epsilon, \Gamma_i = 0] - \lim_{\epsilon \rightarrow 0^-} \mathbb{E}[Y_i \mid g_i = \epsilon, \Gamma_i = 0] \right) (1-\mathbb{E}[\Gamma_i | g_i = 0])\\ 
& = \mathbb{E}[Y_i(1) - Y_i(0) \mid g_i = 0, \Gamma_i = 1] \mathbb{E}[\Gamma_i | g_i = 0]- \mathbb{E}[Y_i(1) - Y_i(0) \mid g_i = 0, \Gamma_i = 0](1-\mathbb{E}[\Gamma_i | g_i = 0]) \\
& = \mathbb{P}(\Gamma_i = 1 | g_i = 0) \int_y Y^{(d)} f_{Y^{(d)}| g, \Gamma}(Y^{(d)} |g = 0, \Gamma = 1) dY^{(d)} \\
& - \mathbb{P}(\Gamma_i = 0 | g_i = 0) \int_y Y^{(d)} f_{Y^{(d)}| g, \Gamma}(Y^{(d)} |g = 0, \Gamma = 0) dY^{(d)} \\
& = \frac{\mathbb{P}(\Gamma_i = 1 | g_i = 0)}{f_{g, \Gamma}(g_i=0, \Gamma = 1)} \int_y Y^{(d)} f_{Y^{(d)}, g, \Gamma}(Y^{(d)}, g = 0, \Gamma = 1) dY^{(d)} \\
& - \frac{\mathbb{P}(\Gamma_i = 0 | g_i = 0)}{f_{g, \Gamma}(g_i=0, \Gamma = 0)} \int_y Y^{(d)} f_{Y^{(d)}, g, \Gamma}(Y^{(d)}, g = 0, \Gamma = 0) dY^{(d)} \\
& = \frac{1}{f_{g_\Omega}(0)} \left( \int_y Y^{(d)} f_{Y^{(d)}, g, \Gamma}(Y^{(d)}, g = 0, \Gamma = 1) dY^{(d)} - \int_y Y^{(d)} f_{Y^{(d)}, g, \Gamma}(Y^{(d)}, g = 0, \Gamma = 0) dY^{(d)} \right) \\
& = \frac{1}{f_{g_\Omega}(0)} \bigg( \int_y Y^{(d)} \int_{X: g_\Omega(X) = 0, \Gamma(X) = 1}f_{Y^{(d)}, X}(Y^{(d)}, X) dS_\Omega dY^{(d)} \\
& - \int_y Y^{(d)}  \int_{X: g_\Omega(X) = 0, \Gamma(X) = 0}f_{Y^{(d)}, X}(Y^{(d)}, X) dS_\Omega dY^{(d)} \bigg) \\
& = \frac{1}{f_{g_\Omega}(0)} \Bigg( \int_y \int_{X: g_\Omega(X) = 0, \Gamma(X) = 1} Y^{(d)} f_{Y^{(d)}|X}(Y^{(d)}|X) f_X(x) dS_\Omega dY^{(d)} \\
& - \int_y \int_{X: g_\Omega(X) = 0, \Gamma(X) = 0}  Y^{(d)} f_{Y^{(d)}|X}(Y^{(d)}|X) f_X(x) dS_\Omega dY^{(d)} \Bigg) \\
& = \frac{1}{f_{g_\Omega}(0)} \left( \int_{X: g_\Omega(X) = 0, \Gamma(X) = 1} \tau(x) f_X(x) dS_\Omega - \int_{X: g_\Omega(X) = 0, \Gamma(X) = 0} \tau(x) f_X(x) dS_\Omega \right) \\
& = \frac{1}{f_{g_\Omega}(0)} \int_{x : g_\Omega(x) = 0} |\tau(x)| f_X(x) \, dS_\Omega = \frac{1}{f_{g_\Omega}(0)} \int_{x : g_\Omega(x) = 0} (2\Gamma(X) - 1)\tau(x) f_X(x) \, dS_\Omega
\end{align*}
\end{proof}

\subsubsection{Proof of Theorem 3}
\begin{proof}
Notice the following:
\begin{align*}
& \left[ \lim_{\epsilon \rightarrow 0^+} f_{g_\Omega| \Lambda}(g = \epsilon | \Lambda_i = 1) - \lim_{\epsilon \rightarrow 0^-} f_{g_\Omega | \Lambda}(g = \epsilon | \Lambda_i = 1) \right] \mathbb{P}(\Lambda_i = 1) - \\
& \left[ \lim_{\epsilon \rightarrow 0^+} f_{g_\Omega | \Lambda}(g = \epsilon | \Lambda_i = 0) - \lim_{\epsilon \rightarrow 0^-} f_{g_\Omega | \Lambda}(g = \epsilon | \Lambda_i = 0) \right] \mathbb{P}(\Lambda_i = 0) \\
= & \left[ \lim_{\epsilon \rightarrow 0^+} f_{g_\Omega, \Lambda}(g = \epsilon, \Lambda_i = 1) - \lim_{\epsilon \rightarrow 0^-} f_{g_\Omega, \Lambda}(g = \epsilon, \Lambda_i = 1) \right] - \\
& \left[ \lim_{\epsilon \rightarrow 0^+} f_{g_\Omega, \Lambda}(g = \epsilon, \Lambda_i = 0) - \lim_{\epsilon \rightarrow 0^-} f_{g_\Omega, \Lambda}(g = \epsilon, \Lambda_i = 0) \right] \\
= & \int_{X: g_\Omega(X) = 0, \Lambda(X) = 1} \tau_f(x) dS_\Omega - \int_{X: g_\Omega(X) = 0, \Lambda(X) = 0} \tau_f(x) dS_\Omega \\
= & \int_{X: g_\Omega(X) = 0} |\tau_f(x)| dS_\Omega = \int_{X: g_\Omega(X) = 0} (2\Lambda(X) - 1)\tau_f(x) dS_\Omega 
\end{align*}
\end{proof}

\subsubsection{Proof of Theorem 4}
\begin{proof}
The estimator for each fold is defined as follows:
$$
\hat{\tau}_{\Gamma_{k, p}} = e_1' (\hat{\beta}_{k,p}^+ - \hat{\beta}_{k,p}^-)
$$
To show that the final estimator is consistent and asymptotically normal, it is sufficient to show that each fold-level estimator is consistent and asymptotically normal. Without loss of generality, notice the following about $\hat{\beta}_{k,p}^+$.
\begin{align*}
\hat{\beta}_{k,p}^+ & = (G_p' W_{+,k} G_p)^{-1} G_p' W_{+,k} \hat{Y} \\
& = (G_p' W_{+,k} G_p)^{-1} G_p' W_{+,k} Y^* + (G_p' W_{+,k} G_p)^{-1} G_p' W_{+,k} \eta
\end{align*}
where 
$$
G_p = \begin{pmatrix}
    r_p(g_1)' \\ r_p(g_2)' \\ \vdots
\end{pmatrix}, \quad W_{+,k} = \mathbbm{1}(g_i \geq 0) \mathbbm{1}(i \in I_k) \text{diag}\left(K\left(\frac{g_1}{h_k}\right), K\left(\frac{g_2}{h_k}\right), ... \right), \quad \eta = \begin{pmatrix}
    \hat{\eta}_1 \\ \hat{\eta}_2 \\ \vdots
\end{pmatrix}, \quad Y^* = \begin{pmatrix}
    Y_1^* \\ Y_2^* \\ \vdots
\end{pmatrix}
$$
Notice the following for the first term (according to the results of \textcite{calonico2014robust}):
\begin{align*}
(G_p' W_{+,k} G_p)^{-1} G_p' W_{+,k} Y^* & \xrightarrow{p} \beta^+ \\
\sqrt{n_k h_k}( (G_q' W_{+,k} G_q)^{-1} G_q' W_{+,k} Y^* - \beta^+ ) & \xrightarrow{d} \mathcal{N}(0, V_+)
\end{align*}
where $V_+$ is the asymptotic variance of the local polynomial regression coefficient estimator and
$$
\epsilon = \begin{pmatrix} \epsilon_1 & \epsilon_2, & \dots  \end{pmatrix}', \quad n_k = \sum_{i=1}^n \mathbbm{1}(i \in I_k)
$$
From here, it is sufficient to show that the second term in $\hat{\beta}_{k,p}^+$ is asymptotically negligible for any polynomial order $p \in \{1, 2, ...\}$. With that, by the law of large numbers, $(G_p' W_{+,k} G_p)^{-1} G_p' W_{+,k} \eta = (\Gamma_p^+)^{-1} \frac{1}{n_k h_k} G_p' W_{+,k} \eta + o_p(n_k^{-1}h_k^{-1})$. Now, notice the following:\footnote{$\mathbb{E}[\eta | X]$ can be separated in the second line since sample splitting is used.}
\begin{align*}
\mathbb{E}[G_p' W_{+,k} \eta] & = \mathbb{E}[G_p' W_{+,k} \mathbb{E}[\eta|X]] \\
& = \frac{2}{h_k} \int_{X: g(X) \geq 0} K(g(X)/h) r_p(g(X)/h) \mathbb{E}[\Gamma_n(X) - \Gamma(X) | X] \mathbb{E}[Y(1) | X] f_X(x) dX \\
& = \int_0^1 K(u) r_p(u) \left( \int_{X: g(X) / h = u} m_1(X) dX \right) du \\
& = \int_0^1 K(u) r_p(u) S_1(hu)  du
\end{align*}
where
$$
m_1(X) = 2\mathbb{E}[\Gamma_n(X) - \Gamma(X) | X] \mathbb{E}[Y(1) | X] f_X(x), \quad S_1(z) = \int_{X: g(X) = z} m_1(X) dX 
$$
Using a Taylor expansion of $S_1(hu)$:
\begin{align*}
\mathbb{E}[G_p' W_{+,k} \eta] & = S_1(0) \int_0^1 K(u) r_p(u) du + h_k S_1'(0) \int_0^1 u K(u) r_p(u) du + \frac{h_k^2}{2} S''_1(0) \int_0^1 u^2 K(u) r_p(u) du + O(h_k^2)
\end{align*}
This means:
$$
\mathbb{E}[G_p' W_{-,k} \eta]  = S_0(0) \int_{-1}^0 K(u) r_p(u) du + h_k S_0'(0) \int_{-1}^0 u K(u) r_p(u) du + \frac{h_k^2 S_0''(0)}{2} \int_{-1}^0 u^2 K(u) r_p(u) du + O(h_k^2)
$$
From here, consider the following:
\begin{align*}
& e_1' (\Gamma_p^+)^{-1} \int_0^1 u^j r_p(u) K(u) du  \\
&  e_1' (\Gamma_p^-)^{-1} \int_{-1}^0 u^j r_p(u) K(u) du
\end{align*}
Notice that:
\begin{align*}
\int_0^1 u^j r_p(u) K(u) du  & = \Gamma_p^+ e_{j+1} \\
\int_{-1}^0 u^j r_p(u) K(u) du  & = \Gamma_p^- e_{j+1}
\end{align*}
Therefore, for $j \leq p$:
\begin{align*}
e_1' (\Gamma_p^+)^{-1} \int_0^1 u^j r_p(u) K(u) du & = e_1' (\Gamma_p^+)^{-1} \Gamma_p^+ e_{j+1} \\
e_1' (\Gamma_p^-)^{-1} \int_{-1}^0 u^j r_p(u) K(u) du& = e_1' (\Gamma_p^-)^{-1} \Gamma_p^- e_{j+1}
\end{align*}
This means that the terms above equal 1 if $j=0$ and 0 if $0 < j \leq p$. This implies the following must hold:
\begin{align*}
e_1' (\Gamma_p^+)^{-1} \mathbb{E}[G_p' W_{+,k} \eta] - e_1' (\Gamma_p^-)^{-1} \mathbb{E}[G_p' W_{-,k} \eta] & = S_1(0) - S_0(0) + O(h_k^2)
\end{align*}
From here, the following holds:
$$
S_1(0) - S_0(0) = \int_{X: g(X) = 0} 2\mathbb{E}[Y(1) - Y(0) | X] \mathbb{E}[\Gamma_n(X) - \Gamma(X) | X] f_X(x) dX
$$

When $\mathbb{E}[Y(1) - Y(0) | X] \neq 0$, $\mathbb{E}[\Gamma_n(X) - \Gamma(X) | X] \xrightarrow{p} 0$. More specifically, notice the following:
\begin{align*}
\mathbb{E}[\Gamma_n(X) - \Gamma(X) | X] & = \mathbb{P}(\Gamma_n(X) - \Gamma(X) = 1 | X) - \mathbb{P}(\Gamma_n(X) - \Gamma(X) = -1 | X) 
\end{align*}
Without loss of generality, suppose that $\tau(X) > 0$. Then the following must hold:
\begin{align*}
\mathbb{P}(\Gamma_n(X) - \Gamma(X) = 1 | X) - \mathbb{P}(\Gamma_n(X) - \Gamma(X) = -1 | X) & = \mathbb{P}(\Gamma_n(X) = 2 | X) - \mathbb{P}(\Gamma_n(X) = 0 | X) \\
& = - \mathbb{P}(\Gamma_n(X) = 0 | X) \\
& = - \mathbb{P}(\tau_n(X) < 0 | X) \\
& = - \mathbb{P}(\tau_n(X) - \tau(X) < -\tau(X) | X) \\
& \geq - \mathbb{P}(\tau_n(X) - \tau(X) < -\tau(X) | X) \\
& - \mathbb{P}(\tau_n(X) - \tau(X) > \tau(X) | X) \\
& = -\mathbb{P}(|\tau_n(X) - \tau(X)| > \tau(X) | X) \\
& \geq -\frac{\mathbb{E}[q(|\tau_n(X) - \tau(X)|) | X]}{q(\tau(X))} 
\end{align*}
where the last line holds by the Markov Inequality for some non-decreasing $q(\cdot) > 0$. By assumption 3, 7, and 8, $\mathbb{E}[q(|\tau_n(X) - \tau(X)|) | X]$ exists for some measurable function $q: \mathbb{R}_+ \to \mathbb{R}_+$. Let $a_h \in (0, 1)$. For any bandwidth $h = O(n^{-a_h})$ and $|\tau_n(X) - \tau(X)| = o_p(1)$, there exists a function $q(\cdot)$ such that $\sqrt{n_k h_k} \mathbb{E}[q(|\tau_n(X) - \tau(X)|) | X] \to 0$. Therefore, $\mathbb{E}[\Gamma_n(X) - \Gamma(X) | X]$ is asymptotically negligible for all $x \in \partial \Omega$, regardless of $a_h$. \\

When $\mathbb{E}[Y(1) - Y(0) | X] = 0$, the integrand above equals zero since $\mathbb{E}[\Gamma_n(X) - \Gamma(X) | X]$ is bounded. This means that $S_1(0) - S_0(0) = 0$. Therefore, we can conclude that:
$$
\sqrt{n_k h_k}e_1'( (G_p' W_{+,k} G_p')^{-1} (G_p' W_{+,k} \eta) - (G_p' W_{-,k} G_p')^{-1} (G_p' W_{-,k} \eta) )\xrightarrow{p} 0
$$
This proves that our estimator is consistent and asymptotically normal. 
\end{proof}

\subsubsection{Proof of Theorem 5}
\begin{proof}
To establish consistency and asymptotic normality, it suffices to verify these properties for each individual $\hat{\tau}_{\Lambda_k}$. Note first that the base kernels $K(\cdot)$ and $L(\cdot)$ are symmetric about zero, implying that the quantities $\hat{g}_{i,k}$ entering the kernels may be treated as if they equal $g_{i,k}$ or $g_{i,k}^*$. Moreover, $\hat{g}_{i,k}^j = (g_{i,k}^*)^j$ for any even integer $j$. With that, both the main kernel estimator and the corresponding bias-correction term can be shown to be consistent. \\

For the main kernel estimator, notice the following:
\begin{align*}
\hat{f}_{g^*}(0^+) - \hat{f}_{g^*}(0^-) & = \frac{1}{n_k} \sum_{i\in I_k} \left(\alpha_1 + \alpha_2 \left| \frac{\hat{g}_{i,k}}{h_k} \right| \right) K\left( \frac{\hat{g}_{i,k}}{h_k} \right)\frac{1}{h_k} (2 \mathbbm{1}(\hat{g}_{i,k} > 0) - 1) \\
& = \frac{1}{n_k} \sum_{i \in I_k} \left(\alpha_1 + \alpha_2 \frac{|g_{i,k}|}{h_k} \right) K\left( \frac{g_{i,k}}{h_k} \right)\frac{1}{h_k} (2 \mathbbm{1}(\hat{g}_{i,k} > 0) - 1) \\
\mathbb{E}[\hat{f}_{g^*}(0^+) - \hat{f}_{g^*}(0^-)] & = \mathbb{E}\left[ \left(\alpha_1 + \alpha_2 \frac{|g_{i,k}|}{h_k} \right) K\left( \frac{g_{i,k}}{h_k} \right)\frac{1}{h_k} (2 \mathbbm{1}(\hat{g}_{i,k} > 0) - 1) \right] \\
& = \mathbb{E}\left[ \left(\alpha_1 + \alpha_2 \frac{|g_{i,k}|}{h_k} \right) K\left( \frac{g_{i,k}}{h_k} \right)\frac{1}{h_k} \mathbb{E}[(2 \mathbbm{1}(\hat{g}_{i,k} > 0) - 1) | X_i] \right] \\
& = \frac{1}{h_k} \int_{\text{Supp}(X)} \left(\alpha_1 + \alpha_2 \frac{|g(X)|}{h_k} \right) K\left( \frac{g(X)}{h_k} \right) \mathbb{E}[(2 \mathbbm{1}(\hat{g}(X) > 0) - 1) | X] f_X(x) dX \\
& = \int_0^1 (\alpha_1 + \alpha_2 u)K(u) \left(\int_{X: g(X) = h_k u}  \mathbb{E}[(2 \mathbbm{1}(\hat{g}(X) > 0) - 1) | X] f_X(x) dX \right) du \\
& + \int_{-1}^0 (\alpha_1 - \alpha_2 u)K(u) \left(\int_{X: g(X) = h_k u}  \mathbb{E}[(2 \mathbbm{1}(\hat{g}(X) > 0) - 1) | X] f_X(x) dX \right) du \\
& = \int_0^1 (\alpha_1 + \alpha_2 u)K(u) S_1(h_k u) du + \int_{-1}^0 (\alpha_1 - \alpha_2 u)K(u) S_0(h_k u) du
\end{align*}
where
\begin{align*}
S_1(z) & = \int_{X: g(X) = z}  \mathbb{E}[(2 \mathbbm{1}(\Lambda_n(X) = 1) - 1) | X] f_X^+(x) dX \\
S_0(z) & = \int_{X: g(X) = z}  \mathbb{E}[(2 \mathbbm{1}(\Lambda_n(X) = 0) - 1) | X] f_X^-(x) dX
\end{align*}
which can also be rewritten as:
\begin{align*}
S_1(z) & = \int_{X: g(X) = z}  (2 \mathbb{P}( \Lambda_n(X) = 1 | X) - 1)  f_X^+(x) dX \\
S_0(z) & = \int_{X: g(X) = z} (2 \mathbb{P}( \Lambda_n(X) = 0 | X) - 1) f_X^-(x) dX
\end{align*}

Using Taylor expansions of $S_1(\cdot)$ and $S_0(\cdot)$:
\begin{align*}
\mathbb{E}[\hat{f}_{g^*}(0^+) - \hat{f}_{g^*}(0^-)] & = (S_1(0) + S_0(0)) + \frac{h_k^2}{2} (S_1''(0) + S_0''(0)) \int_0^1 u^2(\alpha_1 + \alpha_2u) K(u) du + O(h_k^2)
\end{align*}
From here, notice the following:
\begin{align*}
S_1(0) + S_0(0) & = \int_{X: g(X) = 0} \bigg[ (2 \mathbb{P}( \Lambda_n(X) = 1 | X) - 1) f_X^+(x) + (2 \mathbb{P}( \Lambda_n(X) = 0 | X) - 1)  f_X^-(x) \bigg] dX
\end{align*}
If $f_X^+(x) = f_X^-(x)$:
\begin{align*}
S_1(0) + S_0(0) & = \int_{X: g(X) = 0} f_X(x) \left[(2 \mathbb{P}( \Lambda_n(X) = 1 | X) - 1) + (2 \mathbb{P}( \Lambda_n(X) = 0 | X) - 1) \right] \\
& = 0
\end{align*}
When $f_X^+(x) \neq f_X^-(x)$, the integrand above is given as follows:
\begin{align*}
2\mathbb{P}(\tau_n(X) \geq 0 | X) f_X^+(x) + 2\mathbb{P}(\tau_n(X) < 0 | X) f_X^-(x) - (f_X^+(x) + f_X^-(x))
\end{align*}
Suppose, without loss of generality, that $\tau_f(x) > 0$. Then the following must hold:
\begin{align*}
\mathbb{P}(\tau_n(X) < 0 | X) & = \mathbb{P}(\tau_n(X) - \tau_f(x) < -\tau_f(x) | X) \\
& \leq \mathbb{P}(\tau_n(X) - \tau_f(x) < -\tau_f(x) | X) + \mathbb{P}(\tau_n(X) - \tau_f(x) > \tau_f(x) | X) \\
& = \mathbb{P}(|\tau_n(X) - \tau_f(x)| > \tau_f(x) | X) \\
& \leq \frac{\mathbb{E}[q(|\tau_n(X) - \tau_f(x)|) | X]}{q(\tau_f(x))}
\end{align*}
where the last step holds by the Markov Inequality for some function $q: \mathbb{R}_+ \to \mathbb{R}_+$ that is measurable and non-decreasing. Since $q(\cdot)$ is arbitrary, this probability converges to zero at an arbitrary rate that can be made faster than any fixed rate $n_k^{-a}$ for fixed $a \in (0, 1)$. Therefore, in this case:
\begin{align*}
2\mathbb{P}(\tau_n(X) \geq 0 | X) f_X^+(x) + 2\mathbb{P}(\tau_n(X) < 0 | X) f_X^-(x) - (f_X^+(x) + f_X^-(x)) & \to 2f_X^+(x) - f_X^+(x) - f_X^-(x) \\
& = f_X^+(x) - f_X^-(x)
\end{align*}
By an analogous argument, if $\tau_f(x) < 0$, then the integrand converges to $-(f_X^+(x) - f_X^-(x))$. With that, the following holds:
$$
S_1(0) + S_0(0) \to \int_{x: g(x) = 0} |\tau_f(x)| dX
$$
Now, for the second term, notice the following:
\begin{align*}
S_1''(0) + S_0''(0) & = \frac{d^2}{dz^2} \int_{x: g(x) = z} 2\mathbb{P}(\tau_n(X) \geq 0 | X) f_X^+(x) + 2\mathbb{P}(\tau_n(X) < 0 | X) f_X^-(x) - (f_X^+(x) + f_X^-(x)) dX \bigg|_{z = 0}
\end{align*}
Using the same arguments from before, the following holds:
$$
S_1''(0) + S_0''(0) \to f_{g^*}''(0^+) - f_{g^*}''(0^-)
$$
To remove this curvature bias term, first notice the following:
\begin{align*}
\hat{f}_{g^*}''(0^+) - \hat{f}_{g^*}''(0^-) & = \frac{1}{n_k h_k^3} \sum_{i \in I_k} L^{(2)}\left(\frac{\hat{g}_{i,k}}{h_k}\right)(2\mathbbm{1}(\hat{g}_{i,k} > 0)-1) \\
& = \frac{1}{n_k h_k^3} \sum_{i \in I_k} L^{(2)}\left(\frac{g_{i,k}}{h_k}\right)(2\mathbbm{1}(\hat{g}_{i,k} > 0)-1) \\
\mathbb{E}[\hat{f}_{g^*}''(0^+) - \hat{f}_{g^*}''(0^-)] & = \frac{1}{h_k^3} \mathbb{E} \left[ L^{(2)}\left(\frac{g_{i,k}}{h_k}\right)\mathbb{E}[(2\mathbbm{1}(\hat{g}_{i,k} > 0)-1) | X_i] \right] \\
& = \frac{1}{h_k^3} \int_{\text{Supp}(X)} L^{(2)}\left(\frac{g_{i,k}}{h_k}\right)\mathbb{E}[(2\mathbbm{1}(\hat{g}(X) > 0)-1) | X] f_X(x) dX \\
& = \frac{1}{h_k^2} \int_{-1}^1 L^{(2)}(u) \left( \int_{X: g(X) = h_k u} \mathbb{E}[(2\mathbbm{1}(\hat{g}(X) > 0)-1) | X] f_X(x) \right) du \\
& = \frac{1}{h_k^2} \int_0^1 L^{(2)}(u)  S_1(h_k u)du + \frac{1}{h_k^2} \int_0^1 L^{(2)}(u)  S_0(h_k u)du
\end{align*}
Using Taylor expansions for $S_1(\cdot)$ and $S_0(\cdot)$ and integration by parts:
$$
\mathbb{E}[\hat{f}_{g^*}''(0^+) - \hat{f}_{g^*}''(0^-)] = S_1''(0) + S_0''(0) + O(h_k)
$$
which shows that the bias correction 
$$
\frac{h_k^2 \mathbb{E}[\hat{f}_{g^*}''(0^+) - \hat{f}_{g^*}''(0^-)]}{2}\int_0^1 \frac{u^2 (\alpha_1 + \alpha_2 u)K(u)}{2!} du = \frac{h_k^2 (S_1''(0) + S_0''(0))}{2}\int_0^1 \frac{u^2 (\alpha_1 + \alpha_2 u)K(u)}{2!} du + O(h_k^3)
$$
is equivalent to the bias of the main estimator. Since the bias-corrected kernel density estimator can be written as the sum of bounded i.i.d random variables, asymptotic normality follows from the Central Limit Theorem.
\end{proof}

\end{document}